\documentclass[%
 reprint,
 amsmath,amssymb,
 aps,
]{revtex4-2}

\usepackage{graphicx}
\usepackage{dcolumn}
\usepackage{bm}

\usepackage{amsmath,amssymb,amsthm,easybmat,verbatim}

\usepackage{color}

\newtheorem{proposition}{Proposition}
\newtheorem{theorem}{Theorem}

\usepackage{physics}
\usepackage[caption=false]{subfig}
\usepackage{hyperref}

\begin{document}

\preprint{APS/123-QED}

\title{Fault Tolerant Quantum Error Mitigation}

\author{Alvin Gonzales$^{1, 3}$}
\email{agonza@siu.edu}
\author{Anjala M Babu$^2$}
\author{Ji Liu$^3$}
\author{Zain Saleem$^3$}
\email{zsaleem@anl.gov}
\author{Mark Byrd$^{2,4}$}
\affiliation{$^1$Intelligence Community Postdoctoral Research Fellowship Program, Argonne National Laboratory, Lemont, IL, USA}
\affiliation{$^2$School of Physics and Applied Physics, Southern Illinois University, Carbondale, IL, USA}
\affiliation{$^3$Mathematics and Computer Science Division, 
Argonne National Laboratory, Lemont, IL, USA}
\affiliation{$^4$School of Computing, Southern Illinois University, Carbondale, IL, USA}


\date{\today}

\begin{abstract}
Typically, fault-tolerant operations and code concatenation are reserved for quantum error correction due to their resource overhead. Here, we show that fault tolerant operations have a large impact on the performance of symmetry based error mitigation techniques. We also demonstrate that similar to results in fault tolerant quantum computing, code concatenation in fault-tolerant quantum error mitigation (FTQEM)  can exponentially suppress the errors to arbitrary levels. For a family of circuits, we provide analytical error thresholds for FTQEM with the repetition code. These circuits include a set of quantum circuits that can generate all of reversible classical computing. The post-selection rate in FTQEM can also be increased by correcting some of the outcomes. Our threshold results can also be viewed from the perspective of quantifying the number of gate operations we can delay checking the stabilizers in a concatenated code before errors overwhelm the encoding. The benefits of FTQEM are demonstrated with numerical simulations and hardware demonstrations.
\end{abstract}

\maketitle
\makeatother

%


%



\section{Introduction}
Quantum error correction (QEC) is necessary to perform arbitrarily long quantum computation \cite{shor_1996FaultTolQuantComp, Aharonov_2008FaultTolQuantCompWithConstErrRate}. Early experiments have demonstrated the ability of codes such as the surface code \cite{AcharyaGoogle_2023SuppQuantErrsByScalingASurfaceCLogicalQubit}, the four-qubit code \cite{Linke_2017FaultTolQuantErrDetect}, and the repetition code \cite{Kelly_2015StPreservByRepErrDetectInASupercondQuantCirc} to improve the preservation of a quantum state. However, QEC typically incurs large resource overheads in terms of the number of qubits and quantum gates. These extra resources often introduce noise that makes the performance worse than that of the original unencoded payload circuit. In lieu of QEC, error mitigation was developed, where the goal is to reduce enough errors to allow the quantum computer to give useful results. Various error mitigation schemes exist such as dynamical decoupling \cite{viola_1998DynamicalSuppOfDecohInTwoStateQuantSys}, zero noise extrapolation \cite{Li_2017EfficientVarQuantSimIncorpActiveErrMinimiz, Temme_2017ErrMitigForShortDepthQuantCirc, Giurgica-Tiron_2020DigitZNEForQEM}, probabilistic error cancelation \cite{Temme_2017ErrMitigForShortDepthQuantCirc, Endo_2018PractQuantErrMitigForNearFutApplications}, symmetry verification \cite{Mcclean_2017HybridQuantClassHierForMitigOfDecohAndDeterOfExcitedStates, Bonet-Monrog_2018LowCostErrMitigBySymmetry, McArdle_2019ErrMitigatedDigitalQuantSim, Cai_2021QuantErrMitigUsingSymmExpansion}, Pauli check sandwiching (PCS) \cite{Debroy_2020ExtendFlagGadgetsForLowOverCircVerif, Gonzales_2023PCS}, virtual distillation \cite{Huggins_2021VirtualDistilForQEM, Koczor_2021ExpErrSuppForNearTermQuantDevices}, and simulated quantum error mitigation \cite{liu_2022classicalSimsAsQuantErrMitigViaCircCut}. Other work combines error mitigation techniques such as zero noise extrapolation with code concatenation \cite{Wahl_2023ZNEOnLogicalQubitsByScalTheECCDist} and probabilistic error cancellation with quantum codes \cite{Piveteau_2021ErrMitigForUnivGatesOnEncodedQubits}. Error mitigation on encoded states has been demonstrated on hardware in \cite{Urbanek_2020ErrDetectOnQuantCompImprovingTheAccuracyOfChemicalCalc}.

There are multiple error mitigation techniques that verify symmetries. The stabilizer group of stabilizer codes \cite{gottesman_1997stabilizerCodes} can be used to validate symmetries and was examined in \cite{gottesman_2016quantumFaultTolInSmallExp} in the context of quantum error correction and demonstrated on hardware in \cite{Vuillot_2018IsErrDetectHelpfulOnIBM5QChips}. Some techniques verify natural symmetries \cite{Mcclean_2017HybridQuantClassHierForMitigOfDecohAndDeterOfExcitedStates, McArdle_2019ErrMitigatedDigitalQuantSim, Bonet-Monrog_2018LowCostErrMitigBySymmetry, Yen_2019ExactAndApproxSymmProjForElectrStrucProbOnAQuantComp, Sagastizabal_2019ExperimentalErrMitigViaSymmVerInAVQE, AruteGoogle_2020HartreeOnASupercondQubitQuantComp, Stanisic_2022ObservingGroundStPropOfTheFermiModelUsingAScalAlgoOnAQuantComp}, while others verify artificial symmetries of the code space \cite{McClean_2020DecodingQuantErrorsWithSubspaceExpansions}. The direct method from these procedures uses the Hadamard test, but by decomposing the symmetry projector into a basis through subspace expansion, we can use destructive measurements at the end of the circuit to simplify the process \cite{McClean_2020DecodingQuantErrorsWithSubspaceExpansions, Cai_2021QuantErrMitigUsingSymmExpansion}. These results were extended in \cite{endo2022_quantumEMForRotationSymmBosonicCodesWithSymmExpan, Tsubouchi_2023VirtualQuantErrDetection}, where they apply the measurements in other parts of the circuit.

In this work, we investigate symmetry-based error mitigation techniques on code spaces.  We demonstrate that fault-tolerant operations provide a dramatic improvement in fidelity compared to non fault-tolerant gates. First, we show that PCS is equivalent to symmetry verification when dealing with encoded states and logical operations. Next, we introduce Fault Tolerant Quantum Error Mitigation (FTQEM) which combines code concatenation and PCS.  In FTQEM we perform quantum error detection only \textit{once} at the end of the circuit. FTQEM shares many similarities with QEC such as the exponential suppression of errors via code concatenation. 

For FTQEM with the repetition code and circuits consisting of only transversal gates and qubit preparations in the $\ket{0}$ state, we derive an analytical error threshold ($p<\frac{1}{et+1}$, where $t$ is the number of transversal operations in the unencoded circuit) below which we can suppress the errors to arbitrary levels. Since the set of transversal operations includes the non-Clifford gate Toffoli, which is universal for reversible classical computing, the family of quantum circuits include a set of quantum circuits that can generate any reversible classical computation. For a given error rate $p$, the FTQEM threshold can be viewed from the perspective of how long we can delay checking the stabilizers before errors overwhelm the code concatenation. Thus, our results also extend the results in \cite{Abu-nada_2017OptTheFreqOfQEC} where they investigated skipping steps in fault tolerant QEC, but without code concatenation.

The repetition code is equivalent to performing one sided PCS-Z checks. Thus, we can perform double sided PCS checks around transversal operations and decode when performing non transversal operations provided that the errors of the sandwiched operations are below the errors introduced by the checks. In the distributed setting, PCS can be used to protect memory and flying qubits. Consequently, our results have implications in the storage and manipulation of entangled states in quantum networks. 

Next, by using the error syndrome and correcting some of the outcomes via post-processing, we are able to alleviate the problem of an exponentially decreasing post-selection rate at the cost of a small decrease in fidelity, typically. After, we analytically show that the direct implementation of the logical Hadamard gate for the repetition code can only increase the error rate with a depolarizing noise model for FTQEM.

Finally, we perform quantum hardware and classical simulations using the repetition code and the Steane code. We demonstrate, as far as we know, for the first time the exponential suppression of errors on quantum hardware for circuits with logical operations via code concatenation with the repetition code. The repetition code is easily implemented without a large overhead because, on most quantum computers, states are initialized to the ground state.

Note that Hadamard + Toffoli is universal for quantum computing \cite{shi2002_BothToffoliAndCXNeedLitHelpToDoUQC}, but the logical Hadamard gate is not transversal for the repetition code and is excluded from the set of operations that we consider in calculating the threshold. This is necessary because the repetition code is not a full code. If we measure the qubits at the end, Z errors at the end are inconsequential  \cite{berg_2022singleshotErrMitigByCoherentPauliChecks}, but during the computation a single physical Z error can transform via the logical Hadamard gate to cause a logical bit flip error \cite{Guilland_2019RepCatQubitsForFTQC}.


Some open problems or difficulties with implementing FTQEM for arbitrary circuits are the difficulty of performing non-transversal gates fault tolerantly and the connectivity of the hardware device. Non-transversal gates are usually probabilistic and thus cause significant noise. Mapping circuits to quantum hardware with limited connectivity poses a problem because swapping qubits drastically degrades the fidelity. However, this is not an issue for fully connected devices.

\section{Background}
\subsection{Pauli Check Sandwiching}
First, consider the Pauli group on $n$ qubits $\mathcal{P}_n=\{I, X, Y, Z\}^{\otimes n}\times \{\pm 1, \pm i\}$. Pauli Check Sandwiching is a technique that verifies symmetries of an ideal unitary payload quantum circuit $U$ \cite{Debroy_2020ExtendFlagGadgetsForLowOverCircVerif, Gonzales_2023PCS}. Typically, it uses the Pauli group to find elements such that
\begin{align}\label{eq:PCS_cond}
    &R_iUL_i=U,
\end{align}
where $L_i,R_i\in \mathcal{P}_n$. Then, we use a pair $L_i$ and $R_i$ to construct left- and right-controlled checks that sandwich the circuit and post select on the zero ancilla measurement outcome. This scheme easily extends to multiple checks by controlling each pair of checks on a different ancilla. We typically refer to each pair of checks as a layer of the scheme. 
    
\subsection{Stabilizer Formalism and Error Correction}
In quantum error correction, we use additional qubits and define a subspace of the overall Hilbert space as the code space $\mathcal{C}$. Codes are typically defined as $[n,k,d]$, which means that we encode $k$ qubits into $n$ qubits and the code has distance $d$ \cite{gottesman_1997stabilizerCodes}. A code that can correct $t$ errors must have a $d\geq 2t+1$ \cite{gottesman_1997stabilizerCodes}.

A quantum state $\rho$ is said to be stabilized by $O_i\in \mathcal{P}_n$ if it satisfies
\begin{align}\label{eq:stabilizerOp}
    O_i\rho O_i^\dagger = \rho.
\end{align} 
The set $\mathcal{S}=\{O_i\}$ satisfying Eq. \eqref{eq:stabilizerOp} is called the stabilizer of $\rho$ and $S$ is an Abelian subgroup of $\mathcal{P}_n$ \cite{gottesman_1997stabilizerCodes, nielsen2011quantumCompAndQuantInfo}.

In stabilizer quantum error correcting codes, the stabilizer $\mathcal{S}$ of the code space is the subgroup of $\mathcal{P}_n$ that stabilizes all states in the code space.

This implies that the basis for the logical states are the simultaneous $+1$ eigenstates of $\mathcal{S}$. Therefore, by measuring the stabilizer group you can detect errors and for sufficiently sized codes correct errors. In code spaces, we refer to post-selection on the stabilizer measurements where no errors are detected as (stabilizer) quantum error detection. Symmetry verification on code spaces is equivalent to quantum error detection. 

\subsection{Logical Operations}
To make the computation more resilient to errors, we perform logical operations directly on the encoded states. For unencoded states, two universal set of gates are Clifford $+$ T and Toffoli + H, that is, \{H, S, CX, T\} and \{CCX, H\}, respectively. One way of defining the logical operations is to copy the transformations of the physical gates on the computational basis states with the logical gates on the logical basis states.

For an $[n,k,d]$ code, the physical implementations of these logical gates can vary because the constraints are only on the code space, which is a subspace of the entire $2^n$ Hilbert space.

A general and possibly more efficient method to determine a logical operation is to use the encoding map given by the unitary $V$. Suppose $U$ is a physical operation that we want to transform into a logical operation on states in the code space $\mathcal{C}$. From the definition of $V$, we have $VU\rho U^\dagger V^\dagger=(VUV^\dagger) V\rho V^\dagger (VU^\dagger V^\dagger)$, where $VU\rho U^\dagger V^\dagger$ is an element of the code space and $\rho$ is the initial state. Since $\rho_l=V\rho V^\dagger\in \mathcal{C}$, $VUV^\dagger$ is the logical operation for $U$.


\subsection{Fault Tolerance}
For encoding to be beneficial, logical operations should not spread errors to too many qubits when a component fails. A logical operation is fault-tolerant if the failure of one of its components results in at most an error on only one physical qubit per logical qubit \cite{preskill_1997faulttolerantQC}. Transversal gates are an important class of logical gates that are trivially fault tolerant. They are applied bitwise so that a failure in a component results in an error in at most one physical qubit per logical qubit \cite{nielsen2011quantumCompAndQuantInfo}. Unfortunately, the Eastin-Knill theorem states that no quantum code can implement a universal gate set transversely \cite{Eastin_2009RestOnTransversalEncodedQuantGateSets}.

For the two-qubit repetition code, the straightforward implementations of the CX and H gates are given in Figs. \ref{fig:logicalCX} and \ref{fig:logicalH}, respectively. The logical CX is transversal and fault-tolerant, but the logical Hadamard is not. For instance, a failure in a CX gate in the logical H can result in errors on both physical qubits. 
    

\begin{figure}
    \centering
    \subfloat[\label{fig:logicalCX}]{\includegraphics[width=0.15\textwidth]{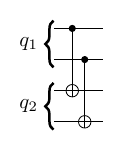}}
    \subfloat[\label{fig:logicalH}]{\includegraphics[width=0.25\textwidth]{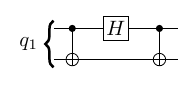}}
    \caption{\ref{fig:logicalCX} is a logical CX for the repetition code. This is transversal and fault tolerant. \ref{fig:logicalH} is a logical H for the repetition code. This is not fault tolerant.}
\end{figure}



\subsection{Avoiding the Hadamard Test}
Typically, measuring the stabilizer group fault-tolerantly is costly because it requires performing a Hadamard test. Direct measurement (DM) \cite{gottesman_2016quantumFaultTolInSmallExp} and decoded stabilizer measurements (DSM) \cite{cory_1998ExpQuantErrCorr} eliminate the need to perform the standard stabilizer (SS) measurements. In DM, the logical states are measured in the computational basis and are decoded by seeing which logical basis state the output string belongs to. If it does not belong to a logical basis state, the result is discarded. In DM, the measurements are transversal and thus DM is a fault-tolerant operation. However, for more complicated codes, this does not verify the phase of the code word.

In DSM, the idea is to decode the logical state and then post-select on the ancillas. The advantage of this method is that it can recover the unmeasured data qubit, and unlike in DM it implicitly post selects based on all the stabilizer generators of the code. From the stabilizer formalism, typically the decoding map transforms all elements of the stabilizer group into the form $I\otimes Z_i \otimes I$, where $Z_i$ acts on the qubit $i$. 
\section{Results}

\subsection{PCS is Equivalent to Stabilizer Quantum Error Detection}
\begin{figure}
    \centering
    \subfloat[\label{fig:1layerPCS}]{\includegraphics[width=0.45\textwidth]{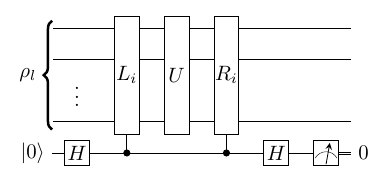}}
    
    \subfloat[\label{fig:stabilErrDetect}]{\includegraphics[width=0.45\textwidth]{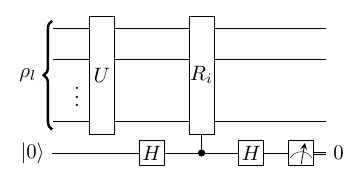}}
   \caption{\ref{fig:1layerPCS} is PCS and \ref{fig:stabilErrDetect} is stabilizer error detection. The two schemes are equivalent provided that we are operating on code spaces and $L_i$ and $R_i$ are elements of the stabilizer group, all satisfying Eq.~\eqref{eq:PCS_cond}. The left check in PCS comes for free because $\rho$ is an element of the code space.}
   \label{fig:pcs_stabilizer_equiv}
\end{figure}

\begin{figure}
    \centering
    \includegraphics[width=0.45\textwidth]{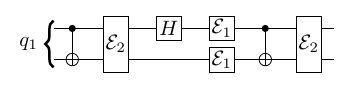}
    \caption{Noisy non fault tolerant Hadamard gate for the repetition code.}
    \label{fig:errorOnNonFTH}
\end{figure}
The first result we show is that PCS is equivalent to stabilizer quantum error detection as depicted in Fig.~\ref{fig:pcs_stabilizer_equiv}. Formally,
\begin{proposition}\label{proposition:PCSEquivToStab}
    Let $\rho$ be an element of the code space $\mathcal{C}$. Then PCS with elements of the stabilizer group $S$ is equivalent to performing stabilizer quantum error detection.
\end{proposition}
\begin{proof}
    Let $U$ be an encoded operation for the quantum error correcting code. Let $L_i$ and $R_i$ be elements of the stabilizer such that $R_iUL_i=U$. The left and right parity checks are
    \begin{align}
        \tilde L_i=\op{0}\otimes\mathbb{I}+\op{1}\otimes L_i
    \end{align}
    and
    \begin{align}
        \tilde R_i=\op{0}\otimes\mathbb{I}+\op{1}\otimes R_i,
    \end{align}
    respectively, where the control is the ancilla and the targets are the data qubits. The full circuit is shown in Fig.~\ref{fig:1layerPCS}.
    

    Since $L_i$ is an element of $\mathcal{S}$, we have
    \begin{align}
        \tilde L_i(\op{+}\otimes \rho)\tilde L_i=\op{+}\otimes \rho, \qquad \forall \rho\in\mathcal{C},
    \end{align}
    where the first system is the ancilla. The left check comes for free because it is equivalent to implementing identity due to $L_i$ stabilizing all the $\rho$ that are in the code space. Thus, PCS is equivalent to stabilizer quantum error detection when used with code spaces and logical operations. 
\end{proof}

Based on previous results in PCS \cite{Gonzales_2023PCS, liu_2022classicalSimsAsQuantErrMitigViaCircCut}, this motivates the idea to perform the stabilizer quantum error detection only once at the end of the circuit.

\subsection{FTQEM Protocol}
FTQEM is essentially equivalent to stabilizer quantum error detection, but we only perform post-selection once at the end of the circuit. Thus, it is almost equivalent to the scheme described in \cite{gottesman_2016quantumFaultTolInSmallExp}, except that we incorporate code concatenation to reduce the errors to arbitrary levels. Thus, the protocol consists of three steps: (1) encode, (2) perform stabilizer quantum error detection once at the end, and (3) concatenate to reduce the noise to arbitrary levels.

\subsection{Concatenation and Bounds on the Performance of FTQEM Repetition Code}
Unlike in quantum error correction, FTQEM allows errors to accumulate. Thus, even when encoding for FTQEM is beneficial, it is not obvious that concatenation will improve things. Here, we prove a theoretical bound on the performance of FTQEM for the repetition code.  




We bound the logical error rates for a family of encoded circuits which consist of only transversal gates. We state the result here and provide the proof in Appendix \ref{appendix:proofThm}. This result shares similarities with the threshold theorem in fault-tolerant quantum error correction \cite{Aharonov_2008FaultTolQuantCompWithConstErrRate, preskill_1997faulttolerantQC}. The main difference is that, since we only perform error detection once at the end of the circuit in FTQEM, the FTQEM threshold is a function of the number of gates. We restrict the operations to $\mathcal{T}=$\{P$_0$, X, S, T, CX, CCX\}, where P$_0$ is the qubit ground state preparation and CCX is Toffoli.


\begin{theorem}\label{thm:thresholdFTQEM}
Let $\mathcal{C}$ be a $[d,k,d]$ repetition code. Let the operations be elements of $\mathcal{T}$.  For an unencoded payload circuit consisting of $t$ total gates, FTQEM can be used to reduce the logical error rate to arbitrary levels, provided that the probability of a physical failure is $p<\dfrac{1}{et+1}$.
\end{theorem}

Since Toffoli is universal for reversible classical computation, Theorem \ref{thm:thresholdFTQEM} includes a set of quantum circuits that can generate any reversible classical computation. Theorem \ref{thm:thresholdFTQEM} can also be viewed from the perspective of how long we can delay checking the stabilizers of the code before too many errors accumulate.

\begin{figure}
    \centering
    \includegraphics[width=0.45\textwidth]{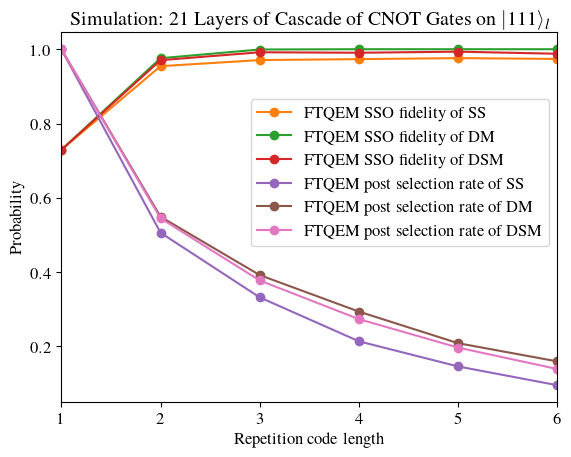}
    \caption{Simulation of a cascade of 21 logical CNOTs on $\ket{111}_l$ logical initial state. The encoding is the repetition code. Noisy two-qubit cat states are used as ancillas for the SS measurements. The single-qubit error rate is 0.1\% and the two-qubit error rate is 1\%. Repetition code length of one means that there is no encoding. The exponential suppression of errors by code concatenation is demonstrated. SS measurements perform the worst in the context of FTQEM.}
    \label{fig:sim_concat_cnot}
\end{figure}

\subsection{Increase Post Selection Rate by Correcting}\label{subsec:correct}
If we implement the current described FTQEM directly, we will get good results, but the post-selection rate decreases exponentially with noise \cite{Gonzales_2023PCS, berg_2022singleshotErrMitigByCoherentPauliChecks}. To address this issue, we can correct some of the outputs with the general cost being a decrease in the fidelity. Let $d$ be the distance of a code. A code can correct errors on $t$ qubits, where $t\leq (d-1)//2$ and $//$ is the floor division. However, we also want to post select away roughly half (in this case 68\%) of the possible incorrect results. 

The possible output strings form a Bell curve distribution with respect to the Hamming weight. For the repetition code, we throw away results that lie within one standard deviation of the mean Hamming weight. This allows the technique to be more scalable. Future research can look at trying other values of the Hamming weight. As the noise levels in quantum computers improve, we can keep more of the results and eventually transition to full quantum error correction. 

\newpage
\onecolumngrid
\begin{widetext}
\begin{center}
\begin{figure}
        \centering
        \subfloat[\label{fig:sim_hgate}]{\includegraphics[width=0.5\textwidth]{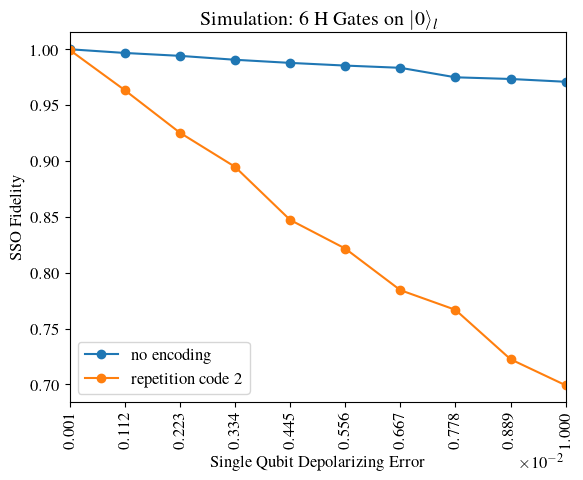}}
        \subfloat[\label{fig:sim_ssdag}]{\includegraphics[width=0.5\textwidth]{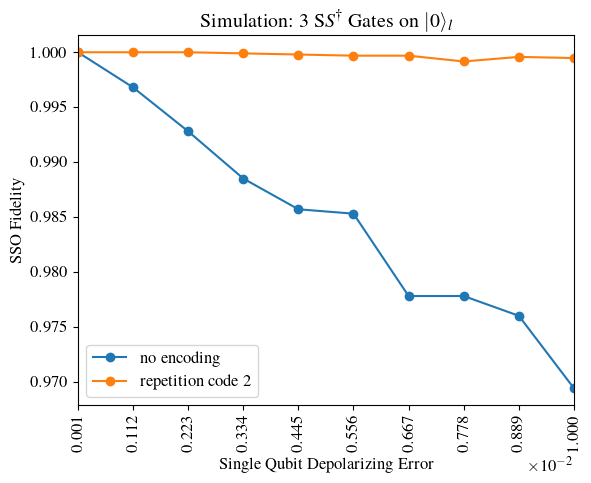}}
       \caption{\ref{fig:sim_hgate} is a simulation consisting of only non fault tolerant H gates repeated six times. \ref{fig:sim_ssdag} is a simulation consisting of a sequence of three SS$^\dagger$ gates. The S gate is transversal and fault tolerant. Encoding improves the performance of the SS$^\dagger$ circuit whereas the performance of the H gate only circuit degrades.}
\end{figure}
\end{center}
\end{widetext}
\twocolumngrid

\subsection{Repetition Code Hadamard}
The straight forward implementation of the logical Hadamard gate in Fig.~\ref{fig:logicalH} is not fault tolerant. Here, we show that the lack of fault tolerance is important in quantum error mitigation with an explicit example. We calculate the logical error rate from depolarizing maps analytically (see Fig.~\ref{fig:errorOnNonFTH}). Let $\mathcal{E}_1$ be the single qubit depolarizing channel, $\mathcal{E}_2$ be the two-qubit depolarizing channel, and $p$ be the depolarizing rate. The logical error rate for an initial starting state $\ket{+}_L$ followed by a noisy logical Hadamard gate is 
\begin{align}
    \dfrac{p(2+(-2+p)p)}{2+2(-1+p)p}.
\end{align}
The error rate for $\ket{+}$ followed by a noisy Hadamard gate is $p/2$. Thus,
\begin{align}
    &\dfrac{p(2+(-2+p)p)}{2+2(-1+p)p}\leq p/2\\
    &\notag\rightarrow\\
    &p\leq 0 \text{ or } 1\leq p
\end{align}
and there is no benefit for using the bit flip code when performing a Hadamard gate. Starting with the basis state $\ket{-}$ shows the same result.


\begin{figure}
    \centering
    \includegraphics[width=0.5\textwidth]{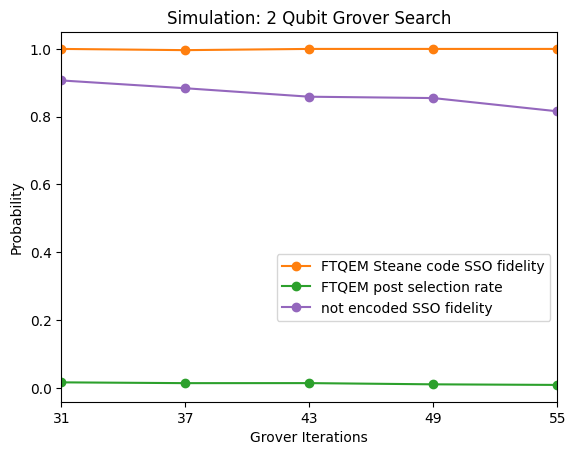}
    \caption{Simulation of two-qubit Grover search with the FTQEM Steane code over varying levels of Grover iterations.}
    \label{fig:SteaneCode}
\end{figure}

\subsection{Simulations}

We use the SSO fidelity $(\sum_i \sqrt{a_ib_i})^2$, where $a_i$ and $b_i$ are the probabilities for outcome $i$ and belong in distributions $a$ and $b$, respectively \cite{Chiaverini_2005SSO}. Unless stated otherwise, we use the repetition code in FTQEM. We choose the repetition code because it is a very light weight code; it typically does not require encoding operations for $\ket{0}_l$, and we can fully verify the symmetry without the Hadamard test by using DM. The repetition code is not a full code because a Z error commutes with its stabilizer group, but we assume that the operations belong in $\mathcal{T}$ and thus Z errors have no effect.

\onecolumngrid
\begin{widetext}    
\begin{center}
\begin{figure}[h!]
        \centering
        \subfloat[\label{fig:hdw_00_cnot11}]{\includegraphics[width=0.5\textwidth]{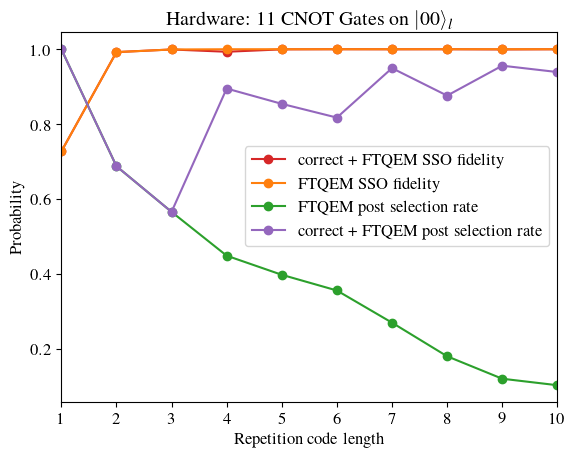}}
        \subfloat[\label{fig:hdw_11_cnot11}]{\includegraphics[width=0.5\textwidth]{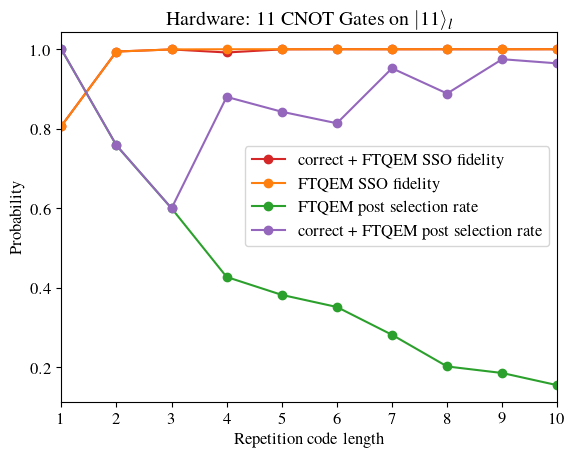}}

        \subfloat[\label{fig:hdw_00_cnot35}]{\includegraphics[width=0.5\textwidth]{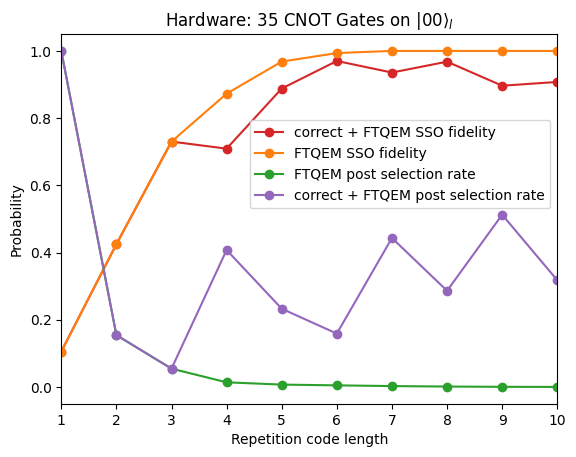}}
        \subfloat[\label{fig:hdw_11_cnot35}]{\includegraphics[width=0.5\textwidth]{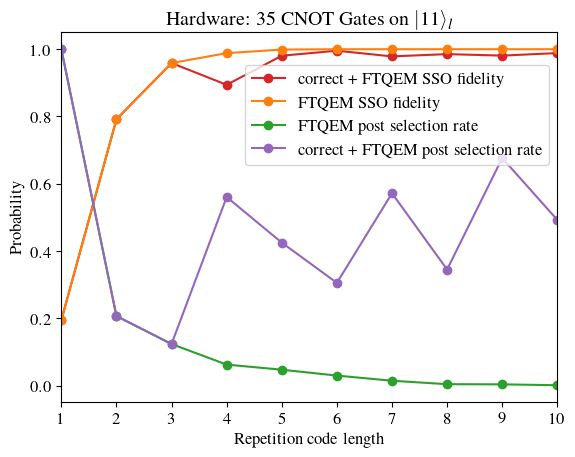}}
       \caption{Hardware demonstrations on ibm\_hanoi. In \ref{fig:hdw_00_cnot11} the circuit starts with the logical $\ket{00}_l$ and executes $11$ CNOT gates. In \ref{fig:hdw_11_cnot11} the circuit starts with the logical $\ket{11}_l$ and executes the $11$ CNOT gates. In \ref{fig:hdw_00_cnot35} the circuit starts with the logical $\ket{00}_l$ and executes the $35$ CNOT gates. In \ref{fig:hdw_11_cnot35} the circuit starts with the logical $\ket{11}_l$ and executes $35$ CNOT gates. A repetition code of length one means no encoding. The executions with ``correct" means that we implemented the correction method discussed in Section \ref{subsec:correct} alongside FTQEM.}
\end{figure}
\end{center}
\end{widetext}
\twocolumngrid

We demonstrate FTQEM with the repetition code under various depolarizing noise and circuits. First, we compare the performance of FTQEM using DM, DSM, or SS measurements. We perform this test using the repetition code and 21 layers of a cascade of logical CNOT gates over the logical $\ket{111}_l$. By varying the repetition length, we also test different levels of concatenation. We use 0.1\% single-qubit depolarizing noise for single-qubit gates and 1\% two-qubit depolarizing noise for two-qubit gates. Repetition code of length one means no encoding. The results are shown in Fig.~\ref{fig:sim_concat_cnot}.

Next, we show the performance of the nonfault-tolerant Hadamard gate and S gate as shown in Figs. \ref{fig:sim_hgate} and \ref{fig:sim_ssdag}, respectively. This corroborates our calculation that the nonfault-tolerant logical Hadamard cannot improve the logical error rate. On the contrary, for the transversal S gate we get an improvement.

Lastly, we demonstrate that FTQEM also works with the Steane code as shown in Fig.~\ref{fig:SteaneCode} for a two-qubit Grover search, which consists of only transversal gates. The error rates used are $0.00003\%$ single-qubit error and $0.002\%$ two-qubit error, which matches rates in H2 Quantinuum devices. These simulations assume full connectivity between qubits.

\subsection{Hardware Demonstrations}
We demonstrate FTQEM on hardware. These results demonstrate exponential suppression of errors through code concatenation on ibm\_hanoi. As mentioned previously, different values of $d$ are equivalent to concatenation of the code, for example, $d=4$ is the concatenation of two $d=2$ encodings. Figs. \ref{fig:hdw_00_cnot11}, \ref{fig:hdw_11_cnot11}, \ref{fig:hdw_00_cnot35}, and \ref{fig:hdw_11_cnot35} demonstrate the performance at various levels of repeated CNOT gates. We also tested the correction method discussed in Section \ref{subsec:correct}. Repetition code of length one means no encoding.

We tested $\ket{00}_l$ and $\ket{11}_l$ initial states with circuits consisting of $11$ and $35$ CNOT gates. We interleaved the two logical qubits to avoid SWAP gates. Note that ibm\_hanoi has an average CNOT error of about $1\%$. Using the threshold formula, $p=0.01<1/(et+1)$, we get $t\leq 36$. Thus, we use 35 CNOTs to make the circuit non-identity. Clearly, we can achieve exponential suppression of errors using FTQEM.

\section{Conclusions}
The performance of FTQEM demonstrates that fault tolerant operations are important in quantum error mitigation. When error detection is performed with non fault-tolerant gates, such as with the non fault-tolerant Hadamard gate in the repetition code, the encoding can make the performance of symmetry-based techniques worse than the unencoded circuit. We show analytically, based on a depolarizing noise model, that the non fault-tolerant Hadamard cannot outperform its unencoded self. In contrast, for a family of fault-tolerant circuits, we demonstrated that concatenation with FTQEM can exponentially decrease the noise to arbitrary levels provided that we are below the FTQEM threshold.

Some interesting open problems are determining the FTQEM threshold for full quantum codes and what the threshold is when the device connectivity requires SWAP gates to be used. For the repetition code, Toffoli and Hadamard may be performed in a error bias preserving manner \cite{Guilland_2019RepCatQubitsForFTQC} via single qubit teleportation \cite{Zhou_2000MethodForQuantLogicGateConst} and the phase basis, but unfortunately the Toffoli implementation requires quantum error correction. It may be possible to replace the error correction with mitigation. Since FTQEM is based on quantum error correction, FTQEM can be used to slowly transition to fault tolerant quantum computing as quantum hardware improves.

\section{Data Availability}
The data presented in this paper is available online at
\url{https://github.com/alvinquantum/FTQEM}.

\section{Code Availability}
The code used for numerical and hardware experiments in this work are available online at \url{https://github.com/alvinquantum/FTQEM}.

\section{Acknowledgements}
This research was supported in part by an appointment to the Intelligence Community Postdoctoral Research Fellowship Program at Argonne National Laboratory, administered by Oak Ridge Institute for Science and Education through an inter-agency agreement between the US Department of Energy and the Office of the Director of National Intelligence. This material is based upon work supported by the U.S. Department of Energy, Office Science, Advanced Scientific Computing Research (ASCR) program under contract number DE-AC02-06CH11357 as part of the InterQnet quantum networking project.  Z.H.S. was supported by the Q-NEXT Center. This research used resources of the Oak Ridge Leadership Computing Facility, which is a DOE Office of Science User Facility supported under Contract DE-AC05-00OR22725.

\appendix

\section{Proof of Theorem \ref{thm:thresholdFTQEM}}\label{appendix:proofThm}

\begin{proof}
We restrict the operations to elements of $\mathcal{T}=$\{P$_0$, X, S, T, CX, CCX\}. Let $a$ denote the unencoded payload circuit, $t$ be the number of gates acting on $a$, and $p$ be the probability of failure of a physical component. Since the gates are transversal, to obtain a logical error, at least $d$ gates must fail out of $td$. Thus, the logical error rate is upper bounded by
\begin{align}
    p_{l}\leq&\sum_{j=d}^{j=td}\binom{td}{j}p^j(1-p)^{td-j}.
\end{align}
We must divide by the post-selection probability to make a valid comparison. The post selection probability is at least the probability of no error
\begin{align}
    p_{s}\geq (1-p)^{td}.
\end{align}
Setting the ratio to less than an arbitrarily small value $\epsilon$, we have
\begin{align}
    &\dfrac{p_{l}}{p_{s}}<\epsilon\\
    \notag&\rightarrow\\
    &\sum_{j=d}^{j=td}\binom{td}{j}p^j(1-p)^{-j}<\epsilon
\end{align}
Using the upper bound on the binomial coefficient $\binom{n}{k}\leq\left(\dfrac{en}{k}\right)^k$ (\cite{odlyzok_1996AsympEnumMethods} equation 4.9), we have
\begin{align}
    \label{eq:summedErr}
    &\sum_{j=d}^{j=td}\binom{td}{j}p^j(1-p)^{-j}\leq\sum_{j=d}^{j=td}\left(\dfrac{etd}{j}\right)^jp^j(1-p)^{-j}<\epsilon.\\
    &\notag\rightarrow\\
    &\label{eq:summedErrRate2}\sum_{j=d}^{j=td}\left(\dfrac{d}{j}\right)^j\left(\dfrac{etp}{1-p}\right)^j<\epsilon
\end{align}

Due to the range of values for $j$, the coefficient $\left(\frac{d}{j}\right)^j\leq 1$ for all $j$. Then, notice that if $\dfrac{etp}{1-p}<1$, Eq.~\eqref{eq:summedErr} can be arbitrarily decreased through $d$. This happens for $p<1/(et+1)$. 

\end{proof}

\bibliographystyle{unsrt}

\vfill

\small

\framebox{\parbox{\linewidth}{
The submitted manuscript has been created by UChicago Argonne, LLC, Operator of 
Argonne National Laboratory (``Argonne''). Argonne, a U.S.\ Department of 
Energy Office of Science laboratory, is operated under Contract No.\ 
DE-AC02-06CH11357. 
The U.S.\ Government retains for itself, and others acting on its behalf, a 
paid-up nonexclusive, irrevocable worldwide license in said article to 
reproduce, prepare derivative works, distribute copies to the public, and 
perform publicly and display publicly, by or on behalf of the Government.  The 
Department of Energy will provide public access to these results of federally 
sponsored research in accordance with the DOE Public Access Plan. 
http://energy.gov/downloads/doe-public-access-plan.}}

\end{document}